\documentclass[a4paper,11pt]{article}
\usepackage[utf8]{inputenc}
\usepackage{authblk}

\pdfoutput=1

\usepackage{times}
\usepackage{bm}
\usepackage{natbib}
\setcitestyle{authoryear}
\usepackage{tikz}
\usepackage{amsmath,amssymb}

\newtheorem{lemma}{Lemma}

\newtheorem{theorem}{Theorem}
\newtheorem{proof}{Proof}

\usepackage{geometry}
\usepackage{color,soul}

\begin{document}
\date{}
\title{Generalizing the Frailty Assumptions in Survival Analysis} 
\author{Vahed Maroufy  \hspace{1cm}{\it and}\hspace{1cm}  Paul Marriott}
\affil{\small Department of Statistics and Actuarial Science, University of Waterloo}
\maketitle

\begin{abstract}
This paper studies Cox's regression hazard model with an unobservable random frailty where no
specific distribution is postulated for the frailty variable, and the marginal lifetime 
distribution allows both parametric and non-parametric models. Laplace's approximation method
and gradient search on smooth manifolds embedded in Euclidean space are applied, and a non-iterative
profile likelihood optimization method is proposed for estimating the regression coefficients.
The proposed method is compared with the Expected-Maximization method developed based on a gamma
frailty assumption, and also in the case when the  frailty model is misspecified.
\end{abstract}


\section{Introduction}\label{Introduction}
Frailty models are important for analyzing survival time data and have been studied by many researchers;
for example, \cite{Klein1992}, \cite{Hougaard1986a}, \cite{Clayton1978} and \cite{Gorfine2006}. One way
of deriving frailty survival models, which we do not follow here, is to formulate the frailty
factor as a single parameter, $\theta$, presenting the association between time-to-event data of two
correlated events, in which $\theta=1$ is interpreted as being no correlation while $\theta>1$ and
$\theta<1$ demonstrate positive and negative association, respectively (\citealp{Hu2011}; \citealp{Nan2006};
\citealp{Clayton1985};\citealp{Clayton1978}; \citealp{Oakes1982}). 

An alternative approach for modeling heterogeneity as unobserved covariate, is to add the frailty variable
as a multiplicative factor to the baseline hazard function. In \cite{Hougaard1986a} a positive stable family
is assumed for the frailty variable and the marginal survival time is assumed to have an exponential or
Weibull  distribution or be unspecified.  Various hazard functions, including Cox's regression model, have
been generalized by assuming a  gamma frailty variable with mean equal to 1 and variance $\eta$, see
\cite{Nielsen1992} and \cite{Klein1992}. They utilize the Expectation-Maximization algorithm for estimation
of parametric and nonparametric accumulated hazard function and regression coefficients. Further,
\cite{Gorfine2006} proposed a different approach for estimation in non-parametric frailty survival models,
which is applicable for any parametric model with finite mean on the frailty variable. 

Although, different models have been assumed for the multiplicative frailty variable, (\citealp{Gorfine2006};
\citealp{Hougaard1984}; \citealp{Hougaard1986a}), one of the most frequently used distributions is the gamma
distribution, because of its tractable properties (\citealp{Klein1992}; \citealp{Nielsen1992}; \citealp{Vaupel1979}).
For example \cite{Martinussen2011} used gamma frailty in the Aalen additive
model, and \cite{Zeng2009} studied transformation models with gamma frailty for multivariate
survival analysis, in which $\eta=0$ (no frailty) is also allowed. In addition, \cite{Abbring2007}
establish the fact that  conditional frailty among survivors is always gamma distributed if and only if the
frailty distribution is regularly varying at zero.

In this paper, we consider Cox's regression model (\citealp{Cox1972}) with a multiplicative frailty factor on which
no specific model is imposed, as biased estimators might be obtained if the frailty model is misspecified
(\citealp{Abbring2007}; \citealp{Hougaard1984}).  Similar to a general mixture model problem, the frailty survival
models with unknown frailty distribution, suffers from identification issues. Although, when all the covariates
variables are continuous with continuous distribution, \citealp{Eleber1982} shows that given the distribution
of the time duration variable, all the three multiplicative factors are identified, his theoretical result does
not solve the identifiability issue in the general sense. For instance, when there is a discrete covariate
then identifiability requires the corresponding regression coefficient to be limited to a known compact set
(\citealp[ch.2]{Horowitz2010}). Consequently, the estimation method developed using unknown transformation
models in \cite{Horowitz1999}, although useful for econometric models, has the same limitation.
In this paper, however, we use the idea of continuous mixture models with relatively small mixing variation when
compared to the total variation. The ``smallness'' restriction allows us to approximate the corresponding mixture
model by a new model which brings nice geometric and inferential properties including identifiability in the general
sense.

This paper is organized as follows. Notation, motivation and the main result of the paper, including our proposed
method, the local mixture method, are presented in  Section \ref{Methodology} for a fixed hazard frailty model. We
estimate the regression coefficients
through a two step optimization process, the first of which is implemented using our proposed algorithm. The algorithm
comprises using a gradient search method on smooth manifolds embedded in finite dimensional Euclidean spaces. The methodology
is generalized to non-parametric baseline hazard in Section \ref{Non-parametric Hazard}. Section
\ref{Simulation} is devoted to simulation studies, illustrating that when frailty is generated from a left-skewed model the
local mixture method returns both smaller bias and standard deviation compared to the existed Expected-Maximization method in
\cite{Klein1992} which assumes a gamma frailty. In Section \ref{Example_survival}, rhDNAse data is analyzed and the results are
compared, for both treatment and placebo group, with the method in \cite{Klein1992}.

\section{Methodology}\label{Methodology}
Throughout this section, we follow the notation and definitions in \cite{Lawless1981} and \cite{Gorfine2006}.
Let $(T^0_{i},C_{i})$, for $i=1,\cdots,n$, be the failure time and censoring time of the $i$th individual,
and also let $X$ be the $n\times p$ design matrix of the covariate vectors. Define $T_i = \min(T^0_{i},C_i)$
and $\delta_i=I(T^0_{i} < C_i)$, where $I(\cdot)$ is an indicator function. In addition, associated with the $i$th
individual, an unobservable covariate $\theta_i$, the frailty, is assumed, where $\theta_i$'s follow some
distribution, $Q$.

Suppose, at least initially, that the marginal lifetime distribution given frailty is an exponential model with
rate $\lambda_0$. Then the baseline hazard function is $\lambda_0(t)=\lambda_0$. Adapting the regression model
in \cite{Cox1972}, the hazard function for the $i$th individual conditional on the frailty $\theta_i$ takes the
following form, 
\begin{eqnarray}
\lambda_i=\theta_i\, \lambda_0 \exp\{X_i\beta^{T} \},\label{Exp_hazard}
\end{eqnarray}
where $X_i$ is the $i$th row of $X$ and $\beta^T=(\beta_0,\cdots,\beta_{p-1})$ is a $p$-vector parameter.
For the $i$th individual with the hazard function in Equation (\ref{Exp_hazard}), the cumulative hazard
function and survival function are, respectively, defined as
\begin{eqnarray}
\Lambda_i(t)=\int_{0}^{t}{\lambda_i(u)\,du},\hspace{1cm} S_i(t)=\exp\{-\Lambda_i(t)\}. 
\end{eqnarray}
Following the arguments in \cite{Gorfine2006}, we assume that the frailty $\theta$ is independent
of $X$, and further that, given $X$ and $\theta$, censoring is independent and noninformative for $\theta$
and $(\lambda_0,\beta)$. Then, for the exponential failure time, the
full likelihood function for the parameter vector $(\lambda_0,\beta)$ is written as
\begin{eqnarray}
 L(\lambda_0,\beta)={\prod_{i=1}^{n}  \int \left[\left(\theta\, \lambda_0\, e^{X_i \beta}\right)^{\delta_i} \exp \left\{-\theta\,T_i \lambda_0\, e^{X_i \beta}\right\}\right] dQ(\theta)  },\label{likelihood_frailty} 
\end{eqnarray}
and the log likelihood function is
\begin{eqnarray}
l(\lambda_0,\beta)=\sum_{i=1}^{n}{\delta_i[\log\lambda_0+ X_i \beta^T]\ + \sum_{i=1}^{n}\log\int{\theta^{\delta_{i}}\,\exp\{-\theta \lambda_0\, T_i \,e^{X_i \beta^T}\}}\, dQ(\theta)}\label{log_like_exp}
\end{eqnarray}

\subsection{Local Mixture Method}\label{Local Mixture Method}
As mentioned in Section 1, it is common in the literature to assume a gamma model,
with $E_Q(\theta)=1$ and variance $\eta$, for $\theta$ and apply the Expectation-Maximization
algorithm for maximizing the log likelihood function in Equation (\ref{log_like_exp}). However,
since frailty model misspecification causes biased coefficient estimation, we relax the gamma
restriction and assume a more general family of distributions for the frailty.
Specifically, $Q$ is supposed to be a proper dispersion model, on an interval
$\Theta$, about the mean value $\vartheta=1$ and with dispersion parameter $\epsilon>0$ (\citealp{Jorgensen1997}). This
assumption allows us to apply Laplace's approximation to the integral in Equation (\ref{log_like_exp}), \cite[ch.6]{Small2010}. In
other words, we only assume, that beyond the observed covariates, there is still a source of
heterogeneity remaining which is unknown and has a relatively small variation about its average
$\vartheta=1$ when compared to the total variation. 
In more detail, let 
$$ f(T_i,\beta,\theta)= \theta^{\delta_{i}}\,\exp\{-\theta \lambda_0 \,T_i\, e^{X_i \beta^T}\},$$ 
then by applying Laplace's expansion to the integral in Equation (\ref{log_like_exp}) as $\epsilon\rightarrow 0$,
we obtain

\begin{eqnarray}
 \int_{\Theta}{f(T_i,\beta,\theta)\,dQ(\theta)}= f(T_i,\beta,\vartheta)+\sum_{j=2}^{k}{\lambda_j\, f^{(j)}(T_i,\beta,\vartheta) }+O\left( \epsilon^{\lfloor \frac{k+1}{2} \rfloor}\right),\label{LMM_exp}
\end{eqnarray}
where $f^{(j)}(T_i,\beta,\vartheta)=\frac{\partial}{\partial \theta^j}\mid_{\theta=\vartheta} f(T_i,\beta,\theta)$,
and $\lambda=(\lambda_2,\cdots,\lambda_k)$ is a parameter vector as a function of $\epsilon$.

The model in Equation (\ref{LMM_exp}) is similar to the local mixture models, introduced in \cite{Marriott2002}
and developed in \cite{Anaya-Izquierdo2007}. For a density function $f$ and a proper dispersion mixing
distribution $Q$, they proposed the local mixture model for studying the  behavior of continuous mixture
models. For any fixed $\vartheta$, the finite dimensional parameter vectors $\lambda$, which represent
the mixing distribution through its central moments, are restricted to a closed convex subspace,
$\Lambda(\vartheta)$. Such a subspace is the intersection of half-spaces containing $\Lambda(\vartheta)$
and therefore bounded by a set of boundary hyper-planes. This boundary, induced by positivity considerations
is  called the hard boundary. 

Local mixing, as defined in \cite{Marriott2002}, extends a parametric model to a larger and more flexible space
of densities which holds nice geometric and inferential properties. Identifiability is achieved by omitting the
first derivative and Fisher orthogonality of the higher derivatives.
It can be shown that this family is richer than the family of mixture models, in the sense that compared to a regular model
with the same mean it can produce both higher and lower dispersion \cite{Marriott2002}. Thus, as shown by
numerical results in Section \ref{Simulation}, local mixtures are quite flexible for modeling unobserved
variation.
Further properties of the local mixture models, such as log concavity of the likelihood, as a function of
$\lambda$ for any fixed $\vartheta$, are studied in \cite{Anaya-Izquierdo2007}. In addition, as they argue,
the notion of ``smallness'' in the local mixture models implies that mixing variation is not the dominant
source of total variation. Hence, in this paper also, frailty is assumed be responsible for a relatively
small part of total variation in the problem, yet remains important from inference point of view.

Substituting Equation (\ref{LMM_exp}) in Equation (\ref{log_like_exp}) we obtain
\begin{eqnarray}
l(\lambda_0,\beta,\lambda)&=&\sum_{i=1}^{n}{\left(\delta_i[\log\lambda_0+ X_i \beta^T]+ \log f(T_i,\beta,\vartheta)\right)}\nonumber\\ 
 &&+  \sum_{i=1}^{n}\log \left( 1+\sum\nolimits_{j=2}^{k}{\lambda_j\, A_{j}(\delta_i,y_i)} \right), \hspace{1cm} \lambda\in \Lambda(\vartheta)\label{log_like_exp1}
\end{eqnarray}
in which $A_{j}(\delta_i,y_i)= \frac{f^{(j)}(T_i,\beta,\vartheta)}{f(T_i,\beta,\vartheta)}$,
and $y_i=\lambda_0 T_i \exp\{X_i \beta^T\}$ is positive. Assuming $\vartheta=1$, i.e., the
average of the frailty distribution is equal to $1$, we maximize Equation (\ref{log_like_exp1}) when 
estimating $\beta$, where $\lambda_0$ and $\lambda$ are consider as nuisance parameters which
are required to be obtained in advance. Thus, a profile likelihood optimization method is
employed. That is, we first maximizes for $\lambda$ over $\Lambda(\vartheta)$ to obtain $\hat{\lambda}$ 
and then maximize $l_p(\beta)=l(\hat{\lambda}_0,\beta,\hat{\lambda})$ to estimate $\beta$. $\hat{\lambda}_0$
is imputed into the loglikelihood function at each iteration. A method for computing $\hat{\lambda}_0$
is described in Section \ref{Non-parametric Hazard}, for a more general situation.

\subsection{Maximum Likelihood Estimator for $\lambda$}\label{Algorithm}

The $\lambda$ parameter space, $\Lambda(\vartheta)$, is characterized as the space of all 
$\lambda$'s such that, for all $y>0$ 
\begin{eqnarray}
1+\sum\nolimits_{j=2}^{k}{\lambda_j\, A_{j}(\delta_i,y)}>0, \hspace{.5cm} \delta_i=0,1 \label{positivity_exp}
\end{eqnarray}
where, $A_{j}(\delta_i,y)$, as a function of $y>0$, is a polynomial of degree $j$. For $k=4$,
the inequality in Equation (\ref{positivity_exp}) is equivalent to the simultaneous positivity conditions
of the following two quartics,
\begin{eqnarray}
p(y)&=& \lambda_4 y^4-\lambda_3 y^3+\lambda_2 y^2+1,\nonumber\\
q(y)&=& \lambda_4 y^4-(4 \lambda_4+\lambda_3) y^3+(3\lambda_3 +\lambda_2)y^2-2\lambda_2 y+1\label{supp_poly}.
\end{eqnarray}
for which we can prove the following result.

\begin{lemma}\label{Lemma1}
If $\Lambda_1$ and $\Lambda_2$ are the space of all $\lambda=(\lambda_2,\lambda_3,\lambda_4)$
such that $p(y)$ and $q(y)$ are positive on $y>0$, respectively, then $\Lambda_2 \subset \Lambda_1$. 
\end{lemma}

\begin{proof}%
First note that $\lambda_4>0$ is a necessary condition hence $P(y)$ has a minimum. Also $q(y)=p(y)-p^{\prime}(y)$,
$p(0)=1$ and $p^{\prime}(0)=0$. 
For all $y>0$, If $q(y)>0$, then $p(y)>p^{\prime}(y)$. Since $p(y)$ attains its minimum value at some $y_1$ for which
$p^{\prime}(y_1)=0$, therefore, $p(y)\geq p(y_1)>p^{\prime}(y_1)=0$. If $y_1=0$, we have $p(y)>p(0)=1$. 
\end{proof}
Lemma \ref{Lemma1} implies that $\Lambda(\vartheta)$ can be characterized just by investigating the positivity domain of
$q(y)$, for which the following theorem is required (see Ulrich \& Watson 1994).

\begin{theorem}\label{Theorem1}
For the quartic polynomial $p(x)=ax^4+bx^3+cx^2+dx+e$, with $a>0$ and $e>0$, define
$\alpha =b\,a^{-3/4}e^{-1/4} , \hspace{.2cm}\beta=c\,a^{-1/2}e^{-1/2} , \hspace{.2cm}\gamma=d\,a^{-1/4}e^{-3/4},$
$$\Delta= 4[\beta^2-3\alpha\gamma+12]^3 -[72\beta+9\alpha \beta \gamma- 2\beta^3 -27\alpha^2 -27\gamma^2]^2$$
\hspace{4cm}$L_1= (\alpha-\gamma)^2-16(\alpha+\beta+\gamma+2)$\\

\hspace{3cm}$L_2 = (\alpha-\gamma)^2-\frac{4(\beta+2)}{\sqrt{\beta-2}}\left(\alpha+\gamma+4\sqrt{\beta-2}\right).$\\
Then, $p(x)\geq0$ for all $x>0$ if and only if
\begin{itemize}
 \item $\beta<-2$ \hspace{.2cm},\hspace{.2cm} $\Delta\leq 0$ \hspace{.2cm},\hspace{.2cm} $\alpha+\gamma>0$
 \item $ -2 \leq\beta\leq 6 \hspace{.2cm},\hspace{.2cm} (\Delta\leq 0 \,\,,\,\,\alpha+\gamma>0) \hspace{.2cm}or\hspace{.2cm} (\Delta\geq 0 \,\,,\,\, L_1 \leq 0)$ 
 \item $ 6 <\beta  \hspace{.2cm},\hspace{.2cm} (\Delta\leq 0 \,\,,\,\,  \alpha+\gamma>0) \hspace{.2cm}or\hspace{.2cm} (\alpha> 0 \,\,,\,\, \gamma>0)  \hspace{.2cm} or \hspace{.2cm}(\Delta\geq 0 \,\,,\,\, L_2 \leq 0)$ 
\end{itemize}
\end{theorem}

Due to the  existence of hard boundaries, obtaining $\hat{\lambda}$ is a nonstandard inference
problem. A suitable maximization algorithm should be flexible enough to converge to a turning
point $\hat{\lambda}$ in the interior if $\hat{\lambda}\in \Lambda(\vartheta)$; otherwise, it
must converge to the unique boundary point with the highest likelihood, say $\hat{\lambda}_b$
(\citealp[p.337]{Berger1987}). In the rest of this section, we propose a gradient based
optimization algorithm, utilizing the geometry of $\Lambda(\vartheta)$ and concavity of the
local mixture term in Equation (\ref{log_like_exp1}) for finding the global maximum value $\hat{\lambda}$
or $\hat{\lambda}_b$ in two major steps. The following lemma reveals the geometry of the boundary
surface of $\Lambda(\vartheta)$, as a smooth manifold embedded in $R^{k-1}$, where $k$ is the
order of the corresponding local mixture model.

\begin{lemma}\label{Lemma2}
The boundary of the parameter space $\Lambda(\vartheta)$, shown by $\Lambda_b(\vartheta)$, is a locally smooth
manifold. 
\end{lemma}
\begin{proof}
see Appendix. 
\end{proof}

\subsection*{Algorithm}

\begin{itemize}
 \item [0:] Start with an initial value $\lambda^{(0)} \in \,\Lambda(\vartheta)$. 
 \item [1:] Run Newton-Raphson algorithm, until either algorithm converges to
                       $\hat{\lambda}\in \Lambda(\vartheta)$ (then stop) or the first update
                       $\lambda^{(j)}\notin \Lambda(\vartheta)$ is obtained (go to step 2). 
 \item [2:] Find the boundary point $\lambda^{\star}$ on the line segment between $\lambda^{(j-1)}$
                       and $\lambda^{(j)}$, let $\lambda^{(j-1)}=\lambda^{\star}$ and run the following steps.
                       \begin{itemize}
                        \item [2a:] Find the gradient $g_j$ and the supporting plane $t_j$ at $\lambda^{(j-1)}$.
                        \item [2b:] Update $\lambda^{(j)}=\lambda^{(j-1)} + (\Pi_j H_j^{-1}) (\Pi_j g_{j})$, (Figure \ref{algorithm}, middle panel, in Appendix).
                        \item [2c:] Update $\lambda^{\star}$ by finding the boundary point on the line segment
                        in the direction of $N_j$, the normal vector of $t_j$, passing through $\lambda^{(j)}$ (Figure \ref{algorithm}, right panel, in Appendix).
                        \item [2d:] Let $\lambda^{(j-1)}=\lambda^{\star}$ and repeat $(2a)$-$(2c)$, until convergence;
                        that is $||P_{t_{j}} \left(g_{j}\right)||<\epsilon$, for a small $\epsilon>0$.
                       \end{itemize}
\end{itemize}
Step 1, obviously applies the well understood Newton-Raphson algorithm on the interior of
$\Lambda(\vartheta)$ as a subspace of $R^{k-1}$. In Step 2, however, a generalization of
Newton's method on smooth manifolds is exploited. Applying Lemma 2 and
using the technical details in Appendix we can prove the following result
(\citealp{Shub1986}; \citealp{Ulrich1994}).

\begin{theorem}\label{Theorem2}
The algorithm either converges to $\hat{\lambda}$ quadratically in step(1), or there
is an open neighborhood $V\subset \Lambda_b(\vartheta)$ of $\hat{\lambda}_b$, that for any
$\lambda^{\star}\in V$ it converges to $\hat{\lambda}_b$ in quadratic order, in step(2).  
\end{theorem}
\begin{proof}
see Appendix.
\end{proof}

\section{Non-parametric Hazard}\label{Non-parametric Hazard}
Although our working example in Section 2 has a fixed hazard rate and exponential lifetime model, the methodology
can be generalized to other parametric marginal lifetime distributions with known hazard function up to a finite
dimensional parameter vector. Furthermore, identifiability property of local mixture models allows the methodology to
be generalized for nonparametric hazard function. When the baseline hazard function is an unknown time-dependent
function $\lambda_0(t)$, the hazard function for $i$th individual takes the form 
\begin{eqnarray}
\lambda_i(t)=\theta_i\, \lambda_0(t) \exp\{X_i \beta^{T}\}.\label{hazard_t}
\end{eqnarray}
The log likelihood function in Equation (\ref{log_like_exp}) has the following form
\begin{eqnarray}
l(\beta,Q)=\sum_{i=1}^{n}{\delta_i[\log\lambda_0(T_i)+ X_i \beta^T]\ + \sum_{i=1}^{n}\log\int{\theta^{\delta_{i}}\,\exp\{-\theta \Lambda_0(T_i) \,e^{X_i \beta^T}\}}\, dQ(\theta)}\nonumber
\end{eqnarray}
and after approximating the integral using a local mixture we obtain 
\begin{eqnarray}
l(\beta,\lambda)&=&\sum\nolimits_{i=1}^{n}{\left(\delta_i[\log\lambda_0(T_i)+ X_i \beta^T]+ \log f(T_i,\beta,\vartheta)\right)}\nonumber\\ 
 &&+  \sum\nolimits_{i=1}^{n}\log \left( 1+\sum\nolimits_{j=2}^{k}{\lambda_j\, A_{j}(\delta_i,y_i)} \right), \hspace{1cm} \lambda\in \Lambda(\vartheta)\label{log_like_exp1np}
\end{eqnarray}
where $y_i=\Lambda_0(T_i) \exp\{X_i \beta^T\}$. Therefore, the geometric and inferential properties of the model
stays the same and we can proceed as in previous section. 

To impute $\lambda_0(t)$ and $\Lambda_0(t)$ we can use the same argument in \cite{Gorfine2006}
to provide a recursive estimate of the cumulative hazard function using the fact that for two consecutive failure times
$T_{(i)}$ and  $T_{(i+1)}$ we have $\Lambda_0(T_{(i+1)})=\Lambda_0(T_{(i)})+\Delta \Lambda_i$. Substituting this recursive
equation in the log likelihood function in (\ref{log_like_exp1np}), considering the conventions in \cite{Breslow1972} and
taking partial derivative with respect to $\Delta \Lambda_i$ we obtain 
\begin{eqnarray}
\frac{\partial l}{\partial \Delta \Lambda_i}=\frac{1}{\Delta \Lambda_i}-\sum_{\ell=i}^{n}e^{X_\ell \beta}+\frac{P^{\prime}(e^{X_i\beta}[\Lambda_0(T_{(i)})+\Delta \Lambda_i])}{P(e^{X_i\beta}[\Lambda_0(T_{(i)})+\Delta \Lambda_i])}\label{partial_diff} 
\end{eqnarray}
which is a function of just $\Delta \Lambda_i$ when $\hat{\Lambda}_0(T_{(i)})$ is given at time $T_{(i+1)}$, where $P(\cdot)$
is a polynomial of degree four with its coefficients as linear functions of $(\lambda_2,\lambda_3,\lambda_4)$ and $P^{\prime}(\cdot)$
is its derivative with respect to $\Delta \Lambda_i$.
When denominator is not zero, equation (\ref{partial_diff}) is a polynomial of degree five which can be solved numerically
for $\Delta \Lambda_i$. Note that when there is no frailty factor; that is, $\lambda=(0,0,0)$ then the last term in equation
(\ref{partial_diff}) is zero, and the estimate of the cumulative hazard function reduces to the form in \cite{Johansen1983}
which is the estimate in \cite{Klein1992} with $\hat{\omega}=1$.


\section{Simulation Study}\label{Simulation}
In this section a simulation study is conducted to compare the local mixture method with the 
method in \cite{Klein1992}, which assumes a gamma model with mean 1 and variance $\eta$, for
the frailty and applies the  Expectation-Maximization algorithm.
Extensive simulation shows that Expectation-Maximization is quite powerful and consistent
as long as the assumptions are not violated. 
However, as shown
in the following, when frailty is generated form a left-skewed model with a small variation
then the local mixture method outperforms the method in \cite{Klein1992}. Also, the Expectation-Maximization method is a
repetitive optimization method while in local mixture method the optimization is performed
in just two steps; hence, it is faster. 

We let $C=0.01$, $\tau=4.6$ and follow a similar set-up as found in \cite{Hsu2004}. For each individual
the event time is $T=[ -\log(1-U) \{  \theta \exp\{\beta X\} \}^{-1}] ^{-1/\tau} C^{-1}$, where
$X\sim N(0,1)$, $U\sim$ uniform$[0,1]$. The censoring distribution is $N(100,15)$, and frailty
is assumed to follow a gamma distribution with mean 1 and variance $\eta$. As shown in Table \ref{simul1},
when $\eta$ is small LMM method works as good as EM method, while for larger $\eta$, EM returns smaller
bias, but LMM method always returns smaller estimation variance.

\begin{table}[h!]
\begin{center}
\caption{\footnotesize $\Gamma(\frac{1}{\eta},\eta)$ .}
\centering 
\begin{tabular}{c c c c c c c c}

\hline 
& & &\multicolumn{2}{c}{LMM } & \multicolumn{2}{c}{EM}\\ [0.3ex]
\hline
$n$ & $\eta$ & $\beta$ &  $bias$ & $std$ & $bias$ & $std$ & iterate \\ [0.3ex]
\hline 
200 & 0.1&  $\log{3}$ & -0.040  & 0.114 & 0.036 & 0.130 & No\\
200 & 0.2&  $\log{3}$ & -0.062  & 0.127 & 0.039 & 0.138 & No\\
200 & 0.4&  $\log{3}$ & -0.073  & 0.129 & 0.011 & 0.175 & No\\
\hline
\end{tabular}
\label{simul1}
\end{center}
\end{table}

Next, we suppose that frailty is misspecified, it
is assumed to follow (i) $Beta(5,1)$ with mean $0.833$ and standard deviation $.141$, (ii) mixture
$0.4 Beta(3,3)+0.6 Beta(3,1)$ with mean $0.65$ and standard deviation $0.225$. 
The bias and standard deviation of the estimates of $\beta$ obtained from $100$
repeated independent samples of size $n$ are reported in Table \ref{Simulation}.
%

\begin{table}[h!]
\begin{center}
\caption{\footnotesize Left skewed frailty; bias and standard deviation (std) of 100 estimates are reported for
both Local Mixture Method (LMM) and Expectation-Maximization (EM).}
\centering 
\begin{tabular}{c c c c c c c}

\hline 
& & & \multicolumn{2}{c}{LMM } & \multicolumn{2}{c}{EM}\\ [0.3ex]
\hline
$n$ & $\beta$ & $model$ &  $bias$ & $std$ & $bias$ & $std$\\ [0.3ex]
\hline 
500 &   log(3) &   (i)  &  -0.0093 &  0.069 &  0.038 &   0.082\\
500 &   log(3) &  (ii)  &  -0.0016 &  0.068 &  0.048 &   0.079\\
\hline
\end{tabular}
\label{simulation}
\end{center}
\end{table}

Table \ref{simulation} shows that for both cases the local mixture model returns both lower bias and lower
standard deviation, where the Expectation-Maximization method returns over estimation in both cases.

\section{Example}\label{Example_survival}
The data was reported based on a clinical trial for assessing the influence of rhDNase on the occurrence
of respiratory exacerbations among patients with cystic fibrosis (\citealp{Fuchs1994}). Among the
645 patients, 324 were assigned to a placebo group using a double-blind randomized design. For both treatment
and placebo group, we study the time to the first occurrence of respiratory exacerbation with two baseline
measurements of forced expository volume, {\em FEV}$_1$ and {\em FEV}$_2 $ as covariates. 
In Table \ref{Example2_2}, the coefficient estimates for the placebo group are reported, and Table \ref{Tab1}
presents the coefficients estimates for the treatment group. The difference between the estimates of the
two methods seems to be negligible for the treatment group, while it is quit noticeable for the placebo group

\begin{table}[h!]
\begin{center}
\caption{\footnotesize Coefficient estimates for placebo group
of rhDNAse data using both methods with unspecified hazard function are obtained.}
\begin{tabular}{c c c c}

\hline 
Method  & $\hat{\beta}_1$ &  $\hat{\beta}_2$\\ [0.3ex]
\hline
EM    &   0.113 & -0.065\\
LMM      &  0.082 & -0.104\\
\hline
\end{tabular}
\label{Example2_2}
\end{center}
\end{table}

\vspace{-.5cm}

\begin{table}[h!]
\begin{center}
\caption{\footnotesize Coefficient estimates for the treatment group of rhDNAse data.}
\begin{tabular}{c c c c}
\hline 
Method  & $\hat{\beta}_1$ &  $\hat{\beta}_2$\\ [0.3ex]
\hline 
EM    &   -0.039 & 0.061\\
LMM      &  -0.040 &  0.060\\
\hline
\end{tabular}
\label{Tab1}
\end{center}
\end{table}

To explore the structural difference between the two data sets we obtain the Poisson process corresponding
to event times for each group by binning the event times. Let 
$$N_j=\sum_{i=1}^n \delta_i I\{T_i\in \xi_j\}, \hspace{1cm} \xi_j=[t_j,t_{j+1})$$
where the length of the binning intervals are assumed to be fixed, $\gamma$ say. For $\gamma\in\{1,
2,\cdots,10\}$
we obtain $10$ different Poisson process for both treatment and placebo group, and compare the ratio of
variance to mean and skewness between the two groups.
\begin{figure}[!h]  
  \center
  \includegraphics[scale=.15]{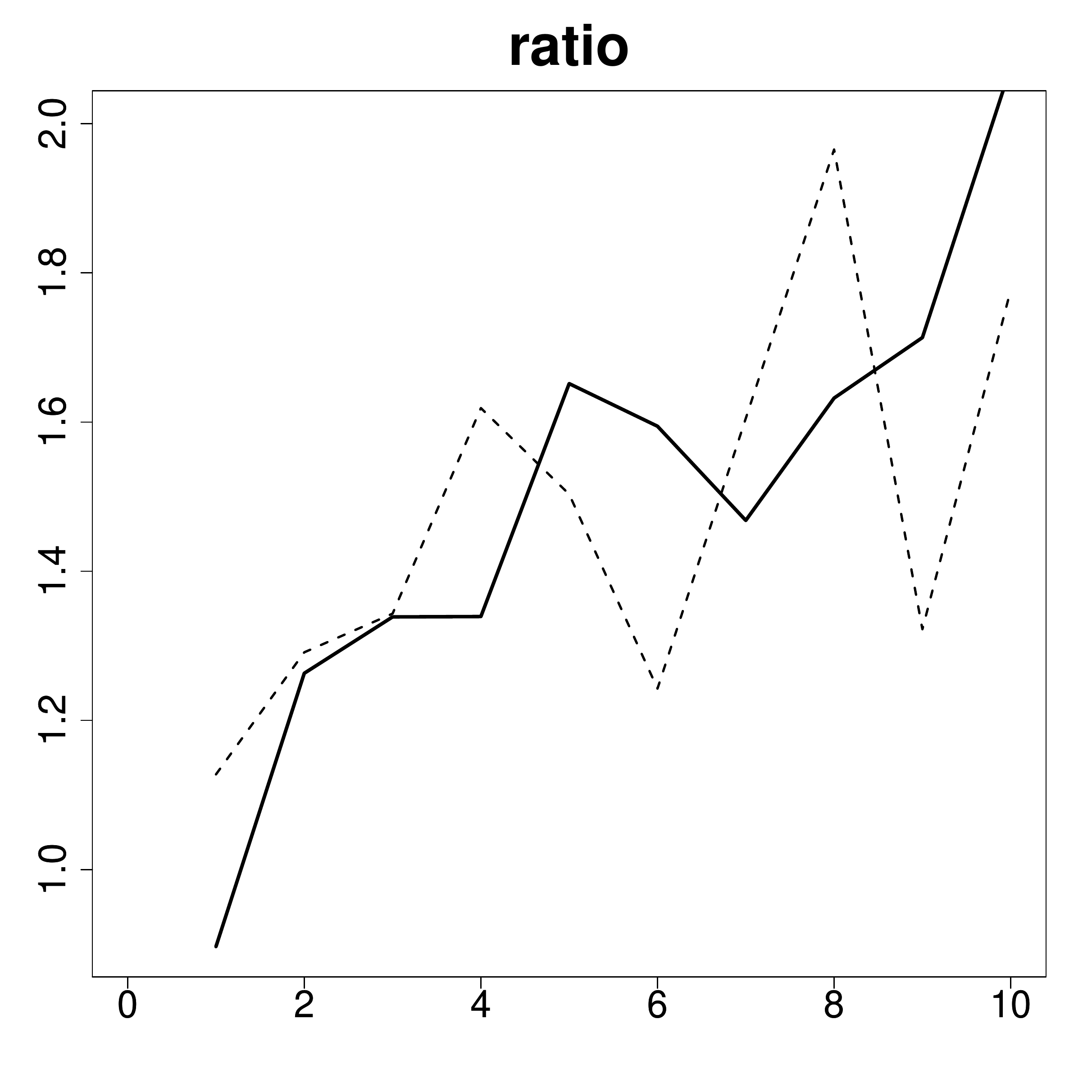} \quad\quad
  \includegraphics[scale=.15]{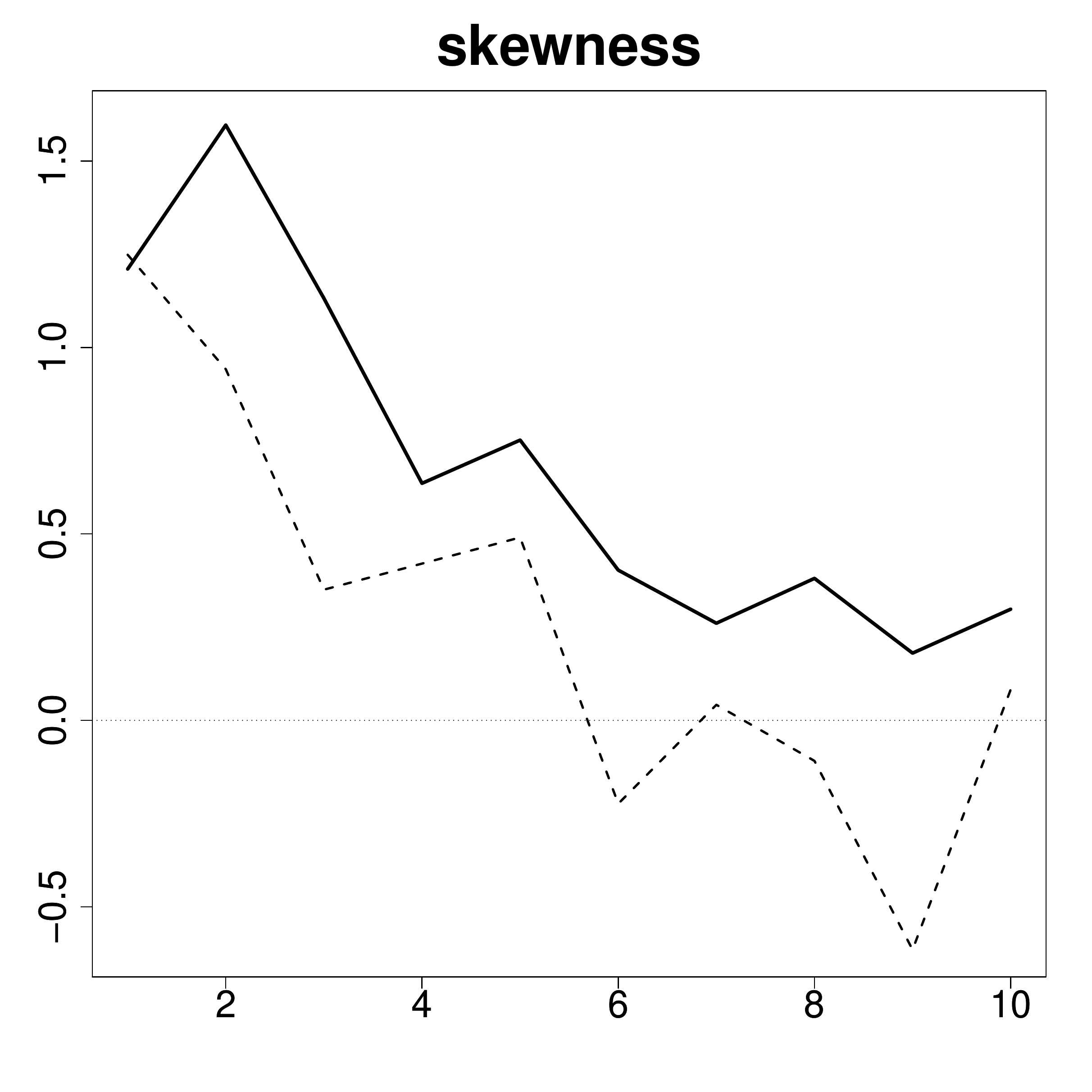}
  \caption{\footnotesize Left: Ratio of variance to mean for 10 different binning of both treatment, solid line,
  and placebo group, dashed line. Right: skewness for 10 different binning of both treatment, solid line,
  and placebo group, dashed line.}\label{meanvar}
\end{figure}

Clearly, for both groups and all the 10 bin-lengths we observe overdispersion, Figure \ref{meanvar}, which
is interpreted as existence of a random effect. However, the plot dose not show any meaningful
difference in the amount of overdispersion between the two groups. Nevertheless, there seems to be a difference
in the skewness structure between the two groups; the treatment group has positive skewness for all 10 bin lengths, while the
placebo group has negative skewness for $\gamma=6,7,8,9$. As illustrated in the simulation study in Section
\ref{Non-parametric Hazard}, when there 
is left-skewness in the random effect the local mixture method returns smaller bias and standard deviation since it is
flexible enough to adjust for it. Therefore, we expect discrepancy in coefficient estimates between the two methods for
placebo group, while they are almost similar for the treatment group.


\bibliographystyle{abbrvnat}
\bibliography{jabref.bib}

\section{Appendix}

\begin{proof}(Lemma \ref{Lemma2})
For $k=4$ (without lose of generality), $\Lambda_b(\vartheta)$ can be parametrized using the solution 
set of equations $q(y)=0$, $q^{\prime}(y)=0$, as functions of $\lambda_2$, $\lambda_3$ and $\lambda_4$
for all $y>0$, which retain the locus of the intersecting line for any two consecutive supporting planes
(Do~Carmo 1976 ch.2). Direct calculation shows that $\Lambda_b(\vartheta)$ can be obtained
by the following smooth mapping,  
$$\mathcal{C}:(0,\infty)\times U\rightarrow R^3, \hspace{1cm} (y,\lambda_2)=\left[\lambda_2,\lambda_3(y,\lambda_2),\lambda_4(y,\lambda_2)\right]$$
where, $U\subset R$ is an open interval and 
$$\lambda_3(y,\lambda_2) = \frac{2(y^3-5y^2+8y)\lambda_2+4y-12}{y^2(y^2-6y+12)} ,\,\,\,\,\,\,\,\,\lambda_4(y,\lambda_2) =\frac{(y^3-4y^2+6y)\lambda_2+3y-6}{y^3(y^2-6y+12)}$$  
Therefore, the implicit function theorem (\citealp[p.224]{Rudin1976}) implies that $\Lambda_b(\vartheta)$ is a smooth manifold.
\end{proof}

In general, the boundary of parameter space of local mixture models may not be smooth manifolds;
hence, in those situations either the possible singularity points must be characterized or other
optimization approaches must be applied for finding maximum on boundary.

\subsection*{Algorithm Description}
To clarify the technical background and convergence proof of the algorithm, the following paragraphs
are in order. For convenience we present the local mixture term in (\ref{log_like_exp1}) by 
$l_{\vartheta}(\lambda)$.\\
In step (2a), $t_j$ is tangent to $\Lambda_b(\vartheta)$ at $\lambda^{(j-1)}=\lambda^{\star}$
and can be obtained as follows. If we collect the quartic $q(y)$ in (\ref{supp_poly}) with
respect to $\lambda_2^{\star}$, $\lambda_3^{\star}$ and  $\lambda_4^{\star}$, we obtain the
supporting plane with the normal vector $( y^{\star2}-2y^{\star},\, -y^{\star3}+3y^{\star2},\, y^{\star4}-4y^{\star3})$,
where $y^{\star}$ is the real multiple root of $q(y)$.

In step (2b), $\Pi_{j}=I-N_j N_j^T$ presents the matrix of orthogonal projection onto tangent
plane $t_j$, with respect to Euclidean inner product, in which $I$ is the identity matrix.
Therefore, for $g_j$ and $H_j$ the gradient vector and hessian matrix of $l_{\vartheta}(\lambda)$
at $\lambda^{(j-1)}$, the first and second covariant derivatives are $\Pi_j g_j$ and $\Pi_j H_j$,
respectively. 

Step (2c), describes the so called exponential -also called retraction- mapping
$R:T\Lambda_b(\vartheta)\rightarrow \Lambda_b(\vartheta)$, where $T\Lambda_b(\vartheta)$ represents
the tangent bundle of $\Lambda_b(\vartheta)$, the disjoint union of $t_j$'s (See Shub 1986).
Let $R_j$ be the restriction of $R$ to $t_j$, then $R_j$ is a one-to-one mapping that maps the vector
$(\Pi_j H_j^{-1})(\Pi_j g_{j})\in t_j$ to a curve between $\lambda^{(j-1)}$ and
$R_j(\lambda^{(j)})$ on $\Lambda_b(\vartheta)$ and holds the following assumptions, 
\begin{itemize}
 \item[1.] $R_j$ is defined in an open interval $U_{r_j}(0_j)\in t_j$, about $0_j$ of  radius $r_j>0$,
 where $0_j$ is the representation of $\lambda^{(j-1)}$ in  $t_j$.
 \item[2.] $R_j(\dot{\lambda})=\lambda$ if and only if $\dot{\lambda}=0_j$.
 \item[3.] $R$ is smooth and $DR_j(0_j)=id_{t_j}$, since $R_j(\lambda^{(j-1)})=\lambda^{(j-1)}$.
\end{itemize}

\begin{figure}[!h] 
  \center
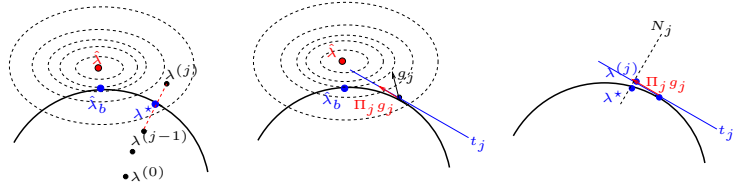
\caption{\footnotesize Schematic visualization of the algorithm steps}\label{algorithm}
\end{figure}

\begin{proof}(Theorem \ref{Theorem2})
Consider the following two cases,\\
(I) $\hat{\lambda} \in \Lambda(\vartheta)$\\
Since $l_{\vartheta}(\lambda)$, for any fixed $\vartheta$, is strictly concave and satisfy the second-order
sufficient conditions, then step (1) of the algorithm converges to the unique global maximum $\hat{\lambda}$
in quadratic order, for any initial point $\lambda^{(0)}$ inside the interior of $\Lambda(\vartheta)$
(see Nocedal \& Wright 2006 p.45).\\
(II) $\hat{\lambda} \notin \Lambda(\vartheta)$\\
Since $\Lambda(\vartheta)$ is closed
and convex in a finite dimensional vector space, there is a unique $\hat{\lambda}_b\in \Lambda(\vartheta)$
with minimum distance from $\hat{\lambda}$, and consequently $l_{\vartheta}(\hat{\lambda}_b)\geq l_{\vartheta}(\lambda)$
for all $\lambda\in \Lambda(\vartheta)$, since $l_{\vartheta}(\lambda)$ is strictly concave. Moreover,
the vector $\hat{\lambda}_b\hat{\lambda}$ is orthogonal to the supporting plane $t_b$, tangent to
$\Lambda(\vartheta)$ at $\hat{\lambda}_b$; hence, $(\Pi_b g_b)$ is a zero vector in the tangent vector
space $t_b$.\\
In addition, according to Lemma 2, $\Lambda_{b}(\vartheta)$, is a locally smooth manifold
embedded in $R^{k-1}$. According to notations in Shub (1986), step 2 can be presented by the following
mapping
\begin{eqnarray}
\mathcal{S}&:&\Lambda_b(\vartheta)\rightarrow \Lambda_b(\vartheta)\nonumber\\
\lambda^{(j-1)}&\rightarrow& R_j\left(\lambda^{(j-1)}\,,\,(\Pi_j H_j^{-1}) (\Pi_jg_{j}) \right) 
\end{eqnarray}
where, by condition (3), $\mathcal{S}$ is smooth. Also, if $(\Pi_j H_j^{-1})$ exists then by conditions (1) and
(2), the fixed points of $\mathcal{S}$ (i.e, $\mathcal{S}(\lambda)=\lambda$) are the zero's of the covariant
gradient, and at fixed points the derivative of $\mathcal{S}$ vanishes.
\end{proof}

\end{document}